\documentclass[12pt,a4paper]{article}

\usepackage{times}
\usepackage[a4paper]{geometry}
\usepackage{graphicx}
\usepackage{amsthm}
\usepackage{amsmath}
\usepackage{amsfonts}
\usepackage{amssymb}
\usepackage{dsfont}
\usepackage{setspace}
\usepackage{algorithm}
\usepackage{algorithmic}
\usepackage{url}
\usepackage{subfigure}
\usepackage{comment}
\usepackage{paralist}
\usepackage{multirow}
\IfFileExists{srcltx.sty} {\usepackage[active]{srcltx}}

\usepackage{graphicx} 
\graphicspath{{./figs/}}
\newcommand{\RR}{\mathbb{R}}
\newcommand{\NN}{\mathbb{N}}
\newcommand{\eqdf}{\triangleq}
\newcommand{\ppath}{\text{path}}
\newcommand{\parent}{\text{parent}}
\newcommand{\depth}{\text{depth}}
\newcommand{\opt}{\text{\textsc{opt}}}
\newcommand{\alg}{\text{\textsc{alg}}}
\newcommand{\algc}{\text{$\alg_c$}}

\newcommand{\algp}{\text{\sc hs$_p$}}
\newcommand{\algd}{\text{\sc hs$_d$}}
\newcommand{\vr}{{{\chi_{\text{\textup{vr}}}}}}

\newcommand{\cmax}{c_{\max}}

\newcommand{\cmin}{c_{\min}}

\newtheorem{definition}{Definition}
\newtheorem{lemma}[definition]{Lemma}
\newtheorem{claim}[definition]{Claim}
\newtheorem{proposition}[definition]{Proposition}
\newtheorem{theorem}[definition]{Theorem}
\newtheorem{corollary}[definition]{Corollary}

\DeclareMathOperator*{\argmin}{arg\,min}
\DeclareMathOperator*{\argmax}{arg\,max}

\begin{document}
\title{Hitting Sets Online and Unique-Max Coloring}
\author{Guy Even\thanks{
School of Electrical Engineering, Tel-Aviv Univ., Tel-Aviv 69978, Israel.}
 \and
Shakhar Smorodinsky
\thanks{
Mathematics Department, Ben-Gurion University of the Negev, Be'er Sheva 84105, Israel.}
}

 \maketitle
  \begin{abstract}
    We consider the problem of hitting sets online. The hypergraph
    (i.e., range-space consisting of points and ranges) is known in
    advance, and the ranges to be stabbed are input one-by-one in an
    online fashion.  The online algorithm must stab each range upon
    arrival. An online algorithm may add points to the hitting set but
    may not remove already chosen points. The goal is to use the
    smallest number of points. The best known competitive ratio for
    hitting sets online  by Alon~et al.~\cite{alon2009online} is
    $O(\log n \cdot \log m)$ for general hypergraphs, where $n$ and
    $m$ denote the number of points and the number of ranges,
    respectively.

    We consider hypergraphs in which the union of two intersecting
    ranges is also a range. Our main result for such hypergraphs is as
    follows. The competitive ratio of the online hitting set problem
    is at most the unique-max number and at least this number minus one.

    Given a graph $G=(V,E)$, let $H=(V,R)$ denote the hypergraph whose
    hyperedges are subsets $U\subseteq V$ such that the induced
    subgraph $G[U]$ is connected. We establish a new connection between
    the best competitive ratio for the online hitting set problem in
    $H$ and the vertex ranking number of $G$.  This connection
    states that these two parameters are equal.  Moreover, this
    equivalence is algorithmic in the sense, that given an algorithm
    to compute a vertex ranking of $G$ with $k$ colors, one can use
    this algorithm as a black-box in order to design a $k$-competitive
    deterministic online hitting set algorithm for $H$.  Also, given a deterministic
    $k$-competitive online algorithm for $H$, we can use it as a black
    box in order to compute a vertex ranking for $G$ with at most $k$
    colors. As a corollary, we obtain optimal online hitting set
    algorithms for many such hypergraphs including those realized by
    planar graphs, graphs with bounded tree width, trees, etc. This
    improves the best previously known general bound of Alon et al.
    \cite{alon2009online}.

    We also consider two geometrically defined hypergraphs. The first
    one is defined by subsets of a given set of $n$ points in the
    Euclidean plane that are induced by half-planes. We obtain an
    $O(\log n)$-competitive ratio. We also prove an $\Omega(\log n)$
    lower bound for the competitive ratio in this setting.  The second
    hypergraph is defined by subsets of a given set of $n$ points in
    the plane induced by unit discs. Since the number of subsets in
    this setting is $O(n^2)$, the competitive ratio obtained by Alon
    et al. is $O(\log^2 n)$. We introduce an algorithm with $O(\log
    n)$-competitive ratio.  We also show that any online algorithm for
    this problem has a competitive ratio of $\Omega(\log n)$, and
    hence our algorithm is optimal.
  \end{abstract}

\section{Introduction}

In the minimum hitting set problem, we are given a hypergraph $(X,R)$,
where $X$ is the ground set of points and $R$ is a set of hyperedges.
The goal is to find a ``small'' cardinality subset $S\subseteq X$ such that every
hyperedge is stabbed by $S$, namely, every hyperedge has a nonempty
intersection with $S$.

The minimum hitting set problem is a classical NP-hard
problem~\cite{Karp72}, and remains hard even for geometrically induced
hypergraphs (see \cite{hochbaum1985approximation} for several
references). A sharp logarithmic threshold for hardness of
approximation was proved by Feige~\cite{feige}. On the other hand, the
greedy algorithm achieves a logarithmic approximation
ratio~\cite{johnson1974approximation,chvatal1979greedy}. Better
approximation ratios have been obtained for several geometrically
induced hypergraphs using specific properties of the induced
hypergraphs~\cite{hochbaum1985approximation,anil1999covering,ben2005constant}.
Other improved approximation ratios are obtained using the theory of
VC-dimension and
$\varepsilon$-nets~\cite{bronnimann1995almost,even2005hitting,clarkson2007improved}.
Much less is known about online versions of the hitting set problem.

In this paper, we consider an online setting in which the set of
points $X$ is given in the beginning, and the ranges are introduced
one by one.  Upon arrival of a new range, the online algorithm may add
points (from $X$) to the hitting set so that the hitting set also
stabs the new range. However, the online algorithm may not remove
points from the hitting set.  We use the competitive ratio for our analysis, a
classical measure for the performance of online
algorithms~\cite{sleator1985amortized,borodin1998online}.

Alon et al. \cite{alon2009online} considered the online set-cover
problem for arbitrary hypergraphs.  In their setting, the ranges are
known in advance, and the points are introduced one by one. Upon
arrival of an uncovered point, the online algorithm must choose a
range that covers the point. Hence, by interchanging the roles of ranges
and points, the online set-cover problem and the online hitting-set problems are equivalent.
The online set cover algorithm presented by Alon et
al.~\cite{alon2009online} achieves a competitive ratio of $O(\log n
\log m)$ where $n$ and $m$ are the number of points and the number of
hyperedges respectively. Note that if $m\geq 2^{n/\log n}$, the
analysis of the online algorithm only guarantees that the competitive
ratio is $O(n)$; a trivial bound if one range is chosen for each
point.
%


\paragraph{Unique-maximum coloring.}
We consider two types of colorings. A coloring $c:X \rightarrow [0,k]$
is a \emph{unique-max coloring} of a hypergraph $H=(X,R)$ if, for each
range $r\in R$, exactly one point is colored by the color $\max_{x\in
  r} c(x)$ (c.f., \cite{CF-survey}).  A {\em vertex ranking} (also an {\em ordered coloring})
of a graph $G=(V,E)$ is a coloring of the vertices $c:V \rightarrow \{1,\ldots,k\}$
that satisfies the following property. Every simple path, the
endpoints of which have the same color $i$, contains a vertex with a
color greater than $i$~\cite{katchalski1995ordered,Schaffer89}.

\paragraph{Relation between unique-maximum coloring and the competitive ratio.}
We consider the competitive ratio for the hitting set problem as a
property of a hypergraph.  Namely, the competitive ratio of a
hypergraph $H=(X,R)$ is the competitive ratio of the best deterministic
online algorithm for the hitting set problem over $H$.  We say that a
hypergraph is $I$-type if the union of two intersecting ranges is
always a range.  Our main result (Theorem~\ref{thm:main}) shows a new
connection between the competitive ratio of an $I$-type hypergraph $H$
and the minimum number of colors required to color $H$ in a unique-max
coloring. In fact, we present ``black box'' reductions that construct
an online hitting set algorithm from a unique-max coloring, and
vice-versa.

\paragraph{Applications.}
Three applications of the main result are presented.  The first
application is motivated by the following setting.  Consider a
communication network $G=(V,E)$. This network is supposed to serve
requests for virtual private networks (VPNs). Each VPN request is a
subset of vertices that induces a connected subgraph in the network,
and requests for VPNs arrive online.  For each VPN, we
need to assign a server (among the nodes in the VPN) that serves the
nodes of the VPN. Since setting up a server is expensive, the goal is
to select as few servers as possible.

This application can be abstracted by considering hypergraphs $H$ that
are realized as the connected induced subgraphs of a given graph $G$.
This hypergraph captures the online problem in which the adversary
chooses subsets $V' \subseteq V$ such that the induced subgraph
$G[V']$ is connected, and the algorithm must stab the subgraphs.  A
direct consequence of the observation that every unique-max coloring
of $H$ is a vertex ranking of $G$
implies that the
competitive ratio of $H$ equals the vertex ranking number of $G$.
This application leads to improved optimal competitive ratios for
graphs that admit (hereditary) small balanced separators (see
Table~\ref{summary}).

Two more classes of hypergraphs are obtained geometrically as follows.
In both settings we are given a set $X$ of $n$ points in the plane. In
one hypergraph, the hyperedges are intersections of $X$ with half
planes.  In the other hypergraph, the hyperedges are intersections of
$X$ with unit discs.  Although these hypergraphs are not $I$-type, we
present an online algorithm for the hitting set problem for points in
the plane and unit discs (or half-planes) with an optimal competitive
ratio of $O(\log n)$. The competitive ratio of this algorithm improves
the $O(\log ^2 n)$-competitive ratio of Alon et al. by a logarithmic
factor.

An application for points and unit discs is the selection of access
points or base stations in a wireless network. The points model base
stations and the disc centers model clients. The reception range of
each client is a disc, and the algorithm has to select a base station
that serves a new uncovered client. The goal is to select as few base
stations as possible.

\paragraph{Organization.}
Definitions and notation are presented in Section~\ref{sec:prelim}.  In Section~\ref{sec:intervals}, we study the special case of
intervals on a line.  The main result is presented in
Section~\ref{sec:main}.  We apply the main result to hypergraphs
induced by connected subgraphs of a given graph in
Section~\ref{sec:connected}.  An online algorithm for hypergraphs
induced by points and half-planes is presented in
Section~\ref{sec:halfplanes}. An online algorithm for the case of
points and unit discs is presented in Section~\ref{sec:discs}.  We
conclude with open problems.

\section{Preliminaries}\label{sec:prelim}
\paragraph{The online minimum hitting set problem.}
Let $H=(X,R)$ denote a hypergraph, where $R$ is a set of nonempty
subsets of the ground set $X$.  Members in $X$ are referred to as
\emph{points}, and members in $R$ are referred to as
\emph{ranges} (or \emph{hyperedges}). A subset $S\subseteq X$
\emph{stabs} a range $r$ if $S\cap r \neq \emptyset$. A
\emph{hitting set} is a subset $S\subseteq X$ that stabs every
range in $R$. In the minimum hitting set problem, the goal is to find
a hitting set with the smallest cardinality.

In this paper, we consider the following online setting. The adversary
introduces a sequence $\sigma\eqdf \{r_i\}_{i=1}^s$ of ranges.  Let
$\sigma_i$ denote the prefix $\{r_1,\ldots,r_i\}$. The online
algorithm must compute a chain of hitting sets $C_1\subseteq C_2
\subseteq \cdots$ such that $C_i$ is a hitting set with respect to the
ranges in $\sigma_i$. In other words, upon arrival of the range $r_i$,
if $r_i$ is not stabbed by $C_{i-1}$, then the online algorithm adds a
point $x_i\in r_i$ to $C_{i-1}$ so that $C_i \eqdf C_{i-1}\cup \{x_i\}$
stabs all the ranges in $\sigma_i$.  If $C_{i-1}$ stabs the range
$r_i$, then the algorithm need not add a point, and $C_i\eqdf C_{i-1}$.

Fix a hypergraph $H$ and an online deterministic algorithm $\alg$.  The competitive
ratio of the algorithm $\alg$ with respect to $H$ is defined as
follows.  For a finite input sequence $\sigma=\{r_i\}_{i=1}^s$, let
$\opt(\sigma_i)\subseteq X$ denote a minimum cardinality hitting set for
the ranges in $\sigma_i$.  Let $\alg(\sigma) \subseteq X$ denote the
hitting set computed by an online algorithm $\alg$ when the input
sequence is $\sigma$.  Note that the sequence of minimum hitting sets
$\{\opt(\sigma_i)\}_{i=1}^s$ is not necessarily a chain of inclusions.
\begin{definition}\label{def:comp ratio}
  The \emph{competitive ratio} of a deterministic online hitting set algorithm
  $\alg$ is defined by
  \begin{align*}
    \rho_H(\alg) &\eqdf \max_{\sigma} \frac{|\alg(\sigma)|}{|\opt(\sigma)|}.
  \end{align*}

The \emph{competitive ratio}  of the hypergraph $H$ is defined by
  \begin{align*}
    \rho_H &\eqdf \min_{\alg} \rho_H (\alg).
  \end{align*}
\end{definition}
The definition of $\rho_H$ can be viewed as a hypergraph property.  It
equals the best competitive ratio achievable by \emph{any} online
deterministic algorithm with respect to the hypergraph $H$.

\paragraph{$I$-type hypergraphs.}
We now define a notion that captures an important property of the hypergraph of intervals over collinear points.
\begin{definition}
  A hypergraph $H=(X,R)$ is \emph{$I$-type} if it satisfies the following property:
  \begin{align*}
    \forall r_1,r_2\in R ~:~~~ r_1\cap r_2 \neq \emptyset \Rightarrow r_1\cup r_2 \in R.
  \end{align*}
\end{definition}

\paragraph{Unique-max colorings.}

Consider a hypergraph $H=(X,R)$ and a coloring $c:X \rightarrow
\NN$. For a range $r\in R$, let $\cmax (r) \eqdf \max \{c(x) \mid
x\in r\}$. Similarly, $\cmin (r) \eqdf \min \{c(x) \mid
x\in r\}$.
\begin{definition}
  A coloring $c:X \rightarrow \NN$ is a \emph{unique-max coloring} of
  $H=(X,R)$ if, for every range $r\in R$, there is a unique point
  $x\in r$ for which $c(x)=\cmax (r)$.
\end{definition}
Similarly, a coloring is unique-min if, for every range $r$, exactly
one point $x\in r$ is colored $\cmin(r)$.

The \emph{unique-max-chromatic number} of a hypergraph $H$, denoted by  $\chi_{um}(H)$,
is the least integers $k$ for which $H$ admits a unique-maximum coloring that uses only $k$ colors.

\paragraph{Vertex ranking.}
We define a coloring notion for graphs known as vertex ranking~\cite{katchalski1995ordered,Schaffer89}.
\begin{definition}
A \emph{vertex ranking} of a graph $G=(V,E)$ is a
  coloring $c:V \rightarrow \NN$ that satisfies the following
  property.  For every pair of distinct vertices $x$ and $y$ and for
  every simple path $P$ from $x$ to $y$, if $c(x) = c(y)$, then there exists
  an internal vertex $z$ in $P$ such that $c(x) < c(z)$.
\end{definition}

The {\em vertex ranking number} of $G$, denoted $\vr(G)$, is the least
integer $k$ for which $G$ admits a vertex ranking that uses only $k$
colors.

A vertex ranking of a graph $G$ is also a
proper coloring of $G$ since adjacent vertices must be colored by
different colors. On the other hand, a proper coloring is not
necessarily a vertex ranking as is easily seen by considering a path
graph $P_n$. This graph admits a proper coloring with $2$ colors but
this coloring is not a valid vertex ranking. In fact, $\vr(P_n)=
\lfloor \log_2 n \rfloor +1$, as proved in the following proposition
that has been proven several times
\cite{iyer1988optimal,katchalski1995ordered,even2003conflict}.

\begin{proposition}\label{prop:Pn}
$\vr(P_n)=
\lfloor \log_2 n \rfloor +1$.
\end{proposition}
\begin{proof}
  Consider a vertex ranking $c$ of $P_n$ in which the highest
  color is used once to split the path into two disjoint paths as
  evenly as possible.  The number of colors $f(n)$ satisfies the
  recurrence $f(1)=1$ and
\begin{align*}
  f(n) &\leq
1+ f\left(\left\lceil\frac{n-1}{2} \right\rceil\right).
\end{align*}
It is easy to verify that $f(n)\leq 1+\lfloor \log_2 n \rfloor$.  For
a matching lower bound, consider a coloring and the point with the highest color. Note that the highest
color appears uniquely in $P_n$. This point
separates the path into two disjoint paths colored by one color less.
The length of one of these paths must be at least
$\left\lceil\frac{n-1}{2} \right\rceil$.  Hence $f(n)\geq 1+
f(\left\lceil\frac{n-1}{2} \right\rceil)$ and therefore $f(n)\geq 1 + \lfloor
\log_2 n \rfloor$, as required.
\end{proof}

\section{Special Case: Hitting set for Intervals on the Line}\label{sec:intervals}
Consider the hypergraph $H=(X,R)$ of intervals over $n$ collinear
points defined by:
\begin{align*}
  X&\eqdf \{1,2,\ldots, n\}\\
R &\eqdf \{[i,j] \mid 1\leq i \leq j \leq n\}.
\end{align*}

The competitive ratio of the online hitting-set algorithm of Alon et
al.~\cite{alon2009online} for the hypergraph of intervals over $n$
collinear points is $O(\log |X| \cdot \log |R|) = O(\log ^2 n)$.  In
this section we prove a better competitive ratio for this specific
hypergraph.

\begin{proposition}\label{prop:intervals}
  $\rho(H) = \lfloor \log_2 n \rfloor +1$.
\end{proposition}
\begin{proof}
  We begin by proving the lower bound $\rho(H) \geq \lfloor \log_2 n
  \rfloor +1.$ The adversary generates the sequence $\sigma\eqdf
  \{r_i\}$ of ranges to be stabbed.  Let $\{C_i\}_i$ denote the
  chain of hitting sets computed by the algorithm.  The first range
  consists of all the points, namely, $r_1=X$. In every step, the next range
  $r_{i+1}$ is chosen to be a larger interval in $r_i\setminus C_i$,
  namely, $|r_{i+1}| \geq \frac{|r_i|-1}{2}$.  While $r_i$ is not
  empty, the adversary forces the algorithm to stab each range by a
  distinct point.  In fact, the adversary can introduce such
  sequence consisting of at least $\lfloor \log_2 n \rfloor +1$ many ranges.
  Thus, $|C_i| =i$ if $i\leq \lfloor \log_2 n \rfloor +1$.  However,
  $r_1 \supset r_2 \supset \cdots$ is a decreasing chain, and hence,
  $|\opt(\sigma_i)|=1$, and the lower bound follows.

  The upper bound $\rho(H) \leq \lfloor \log_2 n \rfloor +1$ is proved
  as follows.  Let $c(x)$ denote a vertex ranking of the graph $P_n$
  that uses $\lfloor\log_2 n \rfloor +1$ colors (see
  Prop.~\ref{prop:Pn}).  Consider the deterministic hitting-set
  algorithm $\alg_c$ defined as follows.  Upon arrival of an unstabbed
  interval $[i,j]$, stab it by the point $x$ in the interval $[i,j]$
  with the highest color.  Namely $x\eqdf \argmax \{c(k) : i\leq k\leq
  j\}$.

We claim that $\rho_H(\alg_c)\leq 1+\lfloor \log_2 n\rfloor$.  The
proof is based on the following observation. Consider a color $\gamma$
and the subsequence of intervals $\sigma(\gamma)$ that consists of the
intervals $r_i$ in $\sigma$ that satisfy the following two properties:
\begin{inparaenum}[(i)]
\item Upon arrival $r_i$ is unstabbed.
\item Upon arrival of $r_i$, $\alg_c$ stabs
$r_i$ by a point colored $\gamma$.
\end{inparaenum}
We claim that the intervals in $\sigma(\gamma)$ are pairwise disjoint.
Indeed, if two intervals $r_1\neq r_2$ in $\sigma(\gamma)$ intersect,
then the maximum color in $r_1 \cup r_2$ is also $\gamma$, and it
appears twice in $r_1\cup r_2$.  This contradicts the definition of a
vertex ranking because $r_1\cup r_2$ is also an interval.  Thus, the
optimum hitting set satisfies $|\opt(\sigma)| \geq \max_{\gamma}
|\sigma(\gamma)|$.  But $|\alg_c (\sigma)| \leq (1+\lfloor \log_2
n\rfloor) \cdot \max_{\gamma} |\sigma(\gamma)|$, and hence
$\rho_H(\alg_c)\leq 1+\lfloor \log_2 n\rfloor$, as required.
\end{proof}

\section{The Main Result}\label{sec:main}
\begin{theorem}\label{thm:main}
    If a hypergraph $H=(X,R)$ is $I$-type, then $$\chi_{um} (H)-1 \leq  \rho(H) \leq \chi_{um} (H).$$
\end{theorem}
The proof of Theorem~\ref{thm:main} is by black-box reductions. The
first reduction uses the unique-max coloring to obtain an online
algorithm (simply stab a range with the point with the highest color).
The second reduction uses a deterministic online hitting set algorithm
to obtain a unique-max coloring.

We say that a hypergraph $H=(X,E)$ is \emph{separable} if
$\{x\}\in R$, for every $x\in X$.
The proof of the following corollary appears in Section~\ref{sec:coro proof}.
\begin{corollary}\label{coro:tight}
If a hypergraph $H=(X,R)$ is $I$-type and separable, then $  \rho(H) = \chi_{um} (H)$.
\end{corollary}

\subsection{Proof of  $\rho(H) \leq \chi_{um} (H)$}\label{sec:ub}
The proof follows the reduction in the proof of
Prop.~\ref{prop:intervals}.  Let $k=\chi_{um}(H)$ and let $c:
X\rightarrow [1,k]$ denote a unique-max coloring
of $H=(X,R)$.  Consider the deterministic hitting-set algorithm
$\alg_c$ defined as follows.  Upon arrival of an unstabbed range $r\in
R$, stab it by the point $x\in r$ colored $\cmax (r)$.

We claim that $\rho_H(\alg_c)\leq k$.  Fix a sequence $\sigma=\{r_i\}_i$ of
ranges input by the adversary.  For a color $\gamma$, let
$\sigma(\gamma)$ denote the subsequence of $\sigma$ that consists of
the ranges $r_i$ in $\sigma$ that satisfy the following properties:
\begin{inparaenum}[(i)]
\item $r_i$ is unstabbed when it arrives.
\item The first point that $\alg_c$ uses to stab $r_i$ is colored
  $\gamma$.
\end{inparaenum}
The ranges in $\sigma(\gamma)$ are pairwise disjoint.  Indeed, if two
ranges $r_1\neq r_2$ in $\sigma(\gamma)$ intersect, then $r_1 \cup
r_2\in R$. Moreover, the maximum color in $r_1 \cup r_2$ is also
$\gamma$. But the color $\gamma$ appears twice in the range $r_1\cup
r_2$; one point that stabs $r_1$ and another point that stabs $r_2$, a
contradiction. Thus the optimum hitting set satisfies $\opt(\sigma)
\geq \max_{\gamma} |\sigma(\gamma)|$.  But
\[
\alg_c (\sigma)
= \sum_{\gamma=1}^{k} |\sigma(\gamma)| \leq k \cdot \max_{\gamma} |\sigma(\gamma)|.
\]
 and hence $\rho_H(\alg_c)\leq k$, as required.

\subsection{Proof of  $\chi_{um} (H)\leq 1+\rho(H)$}
Let $\alg$ denote a deterministic online hitting set algorithm that
satisfies $\rho_H(\alg) = \rho(H)$.  We use $\alg$ as a ``black box''
to compute a unique-min coloring $c:X \rightarrow [0,\rho(H)]$.  Note
that we compute a unique minimum coloring rather than a unique maximum
coloring; this modification simplifies the presentation. (If $c(x)$ is a unique-min coloring, then
$c'(x)\eqdf\rho(H)-c(x)$ is a unique-max coloring.)

\paragraph{Terminology.}
Let $S\subseteq X$ be a subset of points.
We say that a range $r\in R$ is $S$-\emph{maximal} if no range contained in $S$ strictly contains $r$.
Formally, for every range $r'\in R$, $r\subseteq r' \subseteq S$ implies that $r'=r$.
Given a node $v$ in a rooted tree,
let $\ppath(v)$ denote the path from the root to $v$.
Define $\depth(v)$ to be the distance from the root to $v$.
(The distance of the root to itself is zero.)
The \emph{least common ancestor} of two nodes $u$ and $v$ in a rooted
tree is the node of highest depth in $\ppath(u)\cap \ppath(v)$.

\subsubsection{The Decomposition}
We use $\alg$ to construct a decomposition forest consisting of rooted
trees.  Each node $v$ in the forest is labeled by a range $r_v\in R$ and a
point $x_v \in r_v$.
The decomposition forest is defined inductively as follows.

For each $X$-maximal range in $R$ we associate a distinct root.  The
labels of each root $v$ are defined as follows.  The range $r_v$ is the
$X$-maximal range associated with $v$. The point $x_v \in
r_v$ is the point that $\alg$ uses to stab $r_v$ when the input
sequence consists only of $r_v$.

We now describe the induction step for defining the children of a node
$v$ and its labels $r_v$ and $x_v$.  Let $X(\ppath(v))\eqdf \{x_u\mid
u\in\ppath(v)\}$ denote the sequence of points that appear along the
path from the root to $v$.  Similarly, let $\sigma(\ppath(v))$ denote
the sequence of ranges that appear along $\ppath(v)$.  Let $S\eqdf
r_v\setminus X(\ppath(v))$.  For each nonempty $S$-maximal range $r$, we
add a child $v'$ of $v$ that is labeled by the range $r_{v'}=r$.  The
point $x_{v'}$ is the point $x$ that stabs $r_{v'}$ when $\alg$ is
input the sequence of ranges $\sigma(\ppath(v'))$.  We stop with a
leaf $v$ if there is no range contained in $X\setminus X(\ppath(v))$.

\begin{proposition}
  \label{proposition:chain}
  For every node $v$, the sequence of ranges in $\sigma(\ppath(v))$ is
  a strictly decreasing chain. Namely, if $v$ is a child of $u$ then
  $r_v \subsetneq r_u$.  Moreover, when this sequence is input to
  $\alg$, then each range is unstabbed upon arrival.  Hence, the
  points in $X(\ppath(v))$ are distinct.
\end{proposition}
\begin{proposition}
  \label{proposition:sibling}
  If $v_1$ and $v_2$ are siblings, then the ranges $r_{v_1}$ and
  $r_{v_2}$ are disjoint.
\end{proposition}
\begin{proof}
  Otherwise, since $H$ is $I$-type, $r_{v_1}\cup r_{v_2}$ is a range.
  This range contradicts the $S$-maximality of $r_{v_1}$ and $r_{v_2}$
  for $S\eqdf r_v\setminus X(\ppath(v))$, where $v$ is the parent of
  $v_1$ and $v_2$.
\end{proof}
\begin{proposition}
\label{proposition:disjoint} If $v$ and $u$ are two nodes such
  that $v$ is neither an ancestor or a descendant of $u$, then the
  ranges $r_v$ and $r_u$ are disjoint.
\end{proposition}
\begin{proof}
  For the sake of contradiction, assume that $x\in r_u\cap r_v$.  It
  follows that $u$ and $v$ must belong to the same tree whose
  root is labeled by the $X$-maximal range that contains $x$.  The least
  common ancestor $w$ of $u$ and $v$ has two distinct children $w_1$
  and $w_2$ such that $w_1\in \ppath(u)$ and $w_2\in \ppath(v)$. By
  Proposition~\ref{proposition:chain}, $r_u \subset r_{w_1}$ and
  $r_v\subset r_{w_2}$. By Proposition~\ref{proposition:sibling},
  $r_{w_1}\cap r_{w_2} = \emptyset$, and it follows that $r_u\cap
  r_v=\emptyset$, as required.
\end{proof}

\noindent
The following proposition is an immediate consequence of Propositions~\ref{proposition:chain} and~\ref{proposition:disjoint}.
\begin{proposition}
\label{proposition:distinct} All the labels $x_v$ of the nodes
  in the forest are distinct.
\end{proposition}

\subsubsection{Mapping ranges to nodes in the decomposition}
Let $\tilde{X}$ denote the set of nodes in the decomposition forest.
We now define a mapping $f:R \rightarrow \tilde{X}$ from the set of
ranges $R$ to the set of nodes of the decomposition forest.

Define $f(r)$ to be the forest node $v$ of minimum depth such that $x_v$ stabs $r$. Formally,
\begin{align*}
T(r) &\eqdf \{ v\in \tilde{X}\mid  x_v\in r\}\\
f(r)&\eqdf \argmin \{\depth(v) \mid v\in T(r)\}.
\end{align*}

\begin{claim}\label{claim f}
The mapping $f(r)$ is well defined.
\end{claim}
\begin{proof}
We need to prove that (1)~$T(r)$ is not empty for every
  range $r$, and (2)~there exists a unique forest node $v\in T(r)$ of
  minimum depth.

  We prove that $T(r)\neq \emptyset$ by contradiction.  Let $x\in r$
  be any point in $r$. Consider the $X$-maximal range $r_1$ that
  contains $x$.  Let $v_1$ be the root that is associated with $r_1$
  (i.e., $r_{v_1}=r_1$).  Clearly $r\subset r_{v_1}$ because $x\in r$
  and $r_{v_1}$ is $X$-maximal.  By the assumption, $x_{v_1}\notin r$.
  Proceed along the tree rooted at $v_1$ to find a tree path
  $v_1,v_2,\ldots, v_k$ such that $r\subset r_{v_i}$ and
  $x_{v_i}\notin r$ for $1\leq i\leq k$.  To obtain a contradiction,
  we claim that one can find such an infinite path since
  $r\subseteq r_{v_k}\setminus\{x_{v_k}\}$.  Indeed, $r$ is contained
  in one of the $S$-maximal ranges for $S\eqdf r_{v_k}\setminus
  X(\ppath (v_k))$. So we can define $v_{k+1}$ to be the child of
  $v_k$ such that $r\subseteq r_{v_{i+k}}$. Since $T(r)$ is empty
  $x_{v_{k+1}}\not\in r$, the node $v_{k+1}$ meets the requirement
  from the next node in the path.  However, by
  Proposition~\ref{proposition:chain}, each tree in the forest is
  finite, a contradiction.

  We prove that there exists a unique forest node $v\in T(r)$ of
  minimum depth. Assume that there are two forest nodes $u$ and $v$ of
  minimum depth such that both $u$ and $v$ are in $T(r)$. By the
  definition of $S(r)$, $x_u\in r$. By the fact that $\depth(u)$ is
  minimum it follows that $r\cap X(\ppath(u))=\{x_u\}$.  Hence, by the
  maximality of $r_u$, it follows that $r\subseteq r_u$. Analogously,
  $r\subseteq r_v$. Hence, $r_u$ and $r_v$ are not disjoint. By
  Proposition~\ref{proposition:disjoint}, $u$ is an ancestor of $v$,
  or vice-versa. This implies that $\depth(u)\neq \depth(v)$, a
  contradiction.
\end{proof}

\subsubsection{The Coloring}
Define the coloring $c:X \rightarrow \NN$ as follows.
For each $x\in X$, if $x=x_v$ for some forest node $v$, then define $c(x)\eqdf \depth(v)$.
If $x$ does not appear as a label $x_v$ of any node in the forest, then $c(x)\eqdf\rho_H(\alg)$.
Note that Proposition~\ref{proposition:distinct} insures that the coloring $c$ is well defined.

\begin{lemma}\label{lemma:depth}
The depth of every forest node is less than $\rho_H(\alg)$.
\end{lemma}
\begin{proof}
  Consider a node $v$ in the decomposition forest.  By
  Proposition~\ref{proposition:chain}, the sequence $\sigma(\ppath(v))$ of
  ranges is a decreasing chain, and when input to $\alg$, each range
  is unstabbed upon arrival.  Therefore the cardinality of the hitting
  set that $\alg(\sigma(\ppath(v)))$ returns equals $1+\depth(v)$. On
  the other hand, $x_v$ stabs all these ranges. Since the competitive ratio
  of $\alg$ with respect to $H$ is $\rho_H(\alg)$, it follows that the
  length of this sequence is not greater than $\rho_H(\alg)$. The
  length of this sequence equals $1+\depth(v)$, and the
  lemma follows.
\end{proof}
\noindent
Lemma~\ref{lemma:depth} implies that the maximum color assigned by $c(x)$ is $\rho_H(\alg)$.
The following lemma implies that $\chi_{um}(H) \leq \rho_H (\alg)+1$.
\begin{lemma}\label{lemma:unique min}
The coloring $c:X\rightarrow [0,\rho_H(\alg)]$ is a unique-min coloring.
\end{lemma}
\begin{proof}
  Fix a range $r$.  By Claim~\ref{claim f}, $f(r)$ is well defined.
  Thus $r$ contains only one point colored $c(r)$, and all the other
  points in $r$ are colored by higher colors.
\end{proof}
We remark that the proof of Lemma~\ref{lemma:unique min} uses
the color $\rho_H(\alg)$ as a "neutral" color that is never used as the
minimum color in a range.

\subsection{Proof of Corollary~\ref{coro:tight}}\label{sec:coro proof}
\begin{lemma}\label{lemma:equality}
If $H$ is separable, then every point appears as a label $x_v$ in the decomposition.
\end{lemma}
\begin{proof}
  If $H$ is separable, then the stopping condition in the construction
  of the decomposition trees is equivalent to $r_v=\{x_v\}$.
  Otherwise, for each each point $x$ in $r_v\setminus \{x_v\}$, the
  range $\{x\}$ excludes the possibility that $v$ is a leaf.  This
  implies that every point appears as a label of a node in the
  decomposition forest, as required.
\end{proof}

\begin{proof}[Proof of Corollary~\ref{coro:tight}]
  A point $x$ is colored $\rho_H(\alg)$ iff no node is labeled by $x$
  in the decomposition forest.  By Lemma~\ref{lemma:equality}, if $H$
  is separable, then every point appears as a label $x_v$ in the
  decomposition.  Thus, the color $\rho_H(\alg)$ is never used by
  $c(x)$. Hence the range of the coloring $c(x)$ is
  $[0,\rho_H(\alg)-1]$ and the number of colors used by $c(x)$ is only
  $\rho_H (\alg)$, as required.
\end{proof}

\section{Online Hitting-Set for Connected Subgraphs}\label{sec:connected}
We consider the following setting of a hypergraph induced by connected
subgraphs of a given graph. Formally, let $G=(V,E)$ be a graph. Let
$H=(V,R)$ denote the hypergraph over the same set of vertices $V$. A
subset $r\subseteq V$ is a hyperedge in $R$ if and only if the
subgraph $G[r]$ induced by $r$ is connected.

\begin{proposition}\label{prop:vr um}
  A coloring $c:V \rightarrow \NN$ is a vertex ranking of $G$ iff it
  is a unique-max coloring of $H$. Hence, $\chi_{um}(H) = \vr(G)$.
\end{proposition}
In particular, every vertex ranking of the path $P_n$ is a unique-max
coloring of the points with respect to intervals.

The following corollary characterizes the competitive ratio for the
online hitting set problem for $H$ in terms of the vertex ranking
number of $G$. In fact, Propositions~\ref{prop:Pn}
and~\ref{prop:intervals} imply the following corollary for special
case of the path $P_n$.
\begin{corollary}\label{coro:connected}
$\rho(H) = \vr(G).$
\end{corollary}
\begin{proof}
Follows from
Coro.~\ref{coro:tight} and the Proposition~\ref{prop:vr um} .
\end{proof}

Corollary~\ref{coro:connected} implies optimal competitive ratios of
online hitting set algorithms for a wide class of graphs that admit
(hereditary) small balanced separators.  For example, consider the
online hitting set problem for connected subgraphs of a given planar
graph.  Let $G$ be a planar graph on $n$ vertices. It was proved in
\cite{katchalski1995ordered} that $\vr(G)= O(\sqrt{n})$. Therefore,
Coro.~\ref{coro:connected} implies that the competitive ratio of our
algorithm for connected subgraphs of planar graphs is $O(\sqrt{n})$.
Corollary~\ref{coro:connected} also implies that this bound is
optimal.  Indeed, it was proved in \cite{katchalski1995ordered} that
for the $l \times l$ grid graph $G_{l \times l}$ (with $l^2$
vertices), $\vr(G_{l \times l}) \geq l$. Hence, for $G_{l\times l}$,
any deterministic online hitting set algorithm must have a competitive
ratio at least $l$. In Table~\ref{summary} we list several important
classes of such graphs.

\begin{table}\centering
\caption{A list of several graph classes with small separators ($n= |V|$)}
\label{summary}
\begin{tabular}{|l|c|c|}
  \hline
  graph $G=(V,E)$  & competitive ratio & previous result \cite{alon2009online} \\
    \hline   \hline

  path $P_n$ & $\lfloor \log_2 n \rfloor +1$ & $O(\log^2 n)$ \\
    \hline

  tree & $O(\log (\text{diameter}(G))$& $O(n)$ \\  \hline

  tree-width $d$ & $O(d \log n)$ & $O(n)$ \\  \hline

  planar graph & $O(\sqrt{n})$ & $O(n)$ \\
  \hline
\end{tabular}
\end{table}

We note that in the case of a star (i.e., a vertex $v$ with $n-1$
neighbors), the number of subsets of vertices that induce a connected
graph is $2^{n-1}$. Hence, the VC-dimension of the hypergraph is
linear. However, the star has a vertex ranking that uses just two
colors, hence, the competitive ratio of our algorithm in this case is
$2$.  This is an improvement over the analysis of the algorithm of
Alon et al.~\cite{alon2009online} which only proves a competitive
ratio of $O(n)$. Thus, our algorithm is useful even in hypergraphs
whose VC-dimension is unbounded.

\section{Points and Half-Planes}\label{sec:halfplanes}

In this section we consider a special instance of the online
hitting set problem for a finite set of points in the plane and
ranges induced by half-planes.

We prove the following results for hypergraphs in which the ground set
$X$ is a finite set of $n$ points in $\RR^2$ and the ranges are all
subsets of $X$ that can be cut off by a half-plane.  Namely,
a subset of points that lie above (respectively, below) a given line $\ell$.

We note that the hypergraph of points and half-planes is not $I$-type.
See Figure~\ref{fig:pts-planes} for an example. Thus,
Theorem~\ref{thm:main} is not immediately applicable.

\begin{figure}
  \centering
  \includegraphics{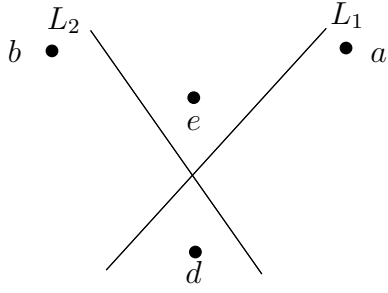}
  \caption{Two intersecting half-planes, the union of which is not
    induced by a half-plane.  The half-planes $r_1=\{a,d\}$ and
    $r_2=\{b,d\}$ intersect. By convexity, every half-plane that contains
    $a,b,$ and $d$ must contain $e$.}
  \label{fig:pts-planes}
\end{figure}
\begin{theorem}\label{thm:LB}
The competitive ratio of every online hitting set algorithm for
points and half-planes is $\Omega(\log n)$.
\end{theorem}

\begin{theorem}\label{thm:UB-points}
There exists an online hitting set algorithm for points and
half-planes that achieves a competitive ratio of $O(\log n)$.
\end{theorem}

In the proofs we consider only ranges of points that are below a line;
the case of points above a line is dealt with separately. This
increases the competitive ratio by at most a factor of two.

\paragraph{Notation.}
Given a finite planar set of points $X$, let $V\subseteq X$
denote the subset of extreme points of $X$. That is, $V$ consists
of all points $p \in X$ such that there exists a half-plane $h$
with $h \cap X = \{p\}$. Let $\{p_i\}_{i=1}^{|V|}$ denote an
ordering of $V$ in ascending $x$-coordinate order. Let
$P=(V,E_P)$ denote the path graph over $V$ where $p_i$ is a
neighbor of $p_{i+1}$ for $i = 1,\ldots, |V|-1$. The intersection
of every half-plane with $V$ is a subpath of $P$. Namely, the
intersection of a nonempty range $r_i$ with $V$ is a set of the
form $\{p_j \mid j\in[a_i,b_i]\}$. We refer to such an
intersection as a discrete interval (or simply an interval, if
the context is clear). We often abuse this notation and refer to
a point $p_i\in V$ simply by its index $i$. Thus, the interval of
points in the intersection of $r_i$ and $V$ is denoted by $I_i
\eqdf [a_i,b_i]$.

\subsection{Proof of Theorem~\ref{thm:LB}}
We reduce the instance of intervals on a line (or equivalently, the
path $P_n$ and its induced connected subgraphs) to an instance of
points and half-planes.  Simply place the $n$ points on the parabola
$y=x^2$. Namely, point $i$ is mapped to the point $(i, i^2)$.  An
interval $[i,j]$ of vertices is obtained by points below the line
passing through the images of $i$ and $j$.  Hence, the problem of
online hitting ranges induced by half-planes is not easier than the
problem of online hitting intervals of $P_n$.  The theorem follows
from Proposition~\ref{prop:intervals}.

\subsection{Proof of Theorem~\ref{thm:UB-points}}
\paragraph{Algorithm Description.}
The algorithm reduces the minimum hitting set problem for points and
half-planes to a minimum hitting set of intervals in a path.  The
reduction is to the path graph $P$ over the extreme points $V$ of $X$.
To apply Algorithm \algc (see Sec.~\ref{sec:ub}), a vertex ranking $c$
for $P$ is computed, and each half-plane $r_i$ is reduced to the
interval $I_i$. A listing of Algorithm \algp\ appears as
Algorithm~\ref{alg:points}. Note that the algorithm \algp\ uses only
the subset $V\subset X$ of extreme points of $X$.

\begin{algorithm}
  \caption{\algp$(\{r_i\})$ - an online hitting set for points and half-planes}
\label{alg:points}
\begin{algorithmic}[1]
  \REQUIRE $X\subset \RR^2$ is a set of $n$ points, and each $r_i$ is
  an intersection of $X$ with a half-plane.

\STATE $V\gets$ the extreme points of $X$ (i.e., lower envelope of the convex hull).

\STATE $\{p_i\}_{i=1}^{|V|}\gets$ ordering of $V$ in ascending $x$-coordinate
order.

\STATE Let $P=(V,E_P)$ denote the path graph over $V$, where $E_P\eqdf
\{(p_i,p_{i+1})\}_{i=1}^{|V|-1}$.
\STATE $c \gets$ a vertex ranking of $P$ (with $\lfloor \log_2 |V|
\rfloor +1$ colors).
\STATE Upon arrival of range $r_i$, reduce it to the interval
$I_i=r_i\cap V$.

\STATE Run $\algc$ with the sequence of ranges $\{I_i\}_i$.
\end{algorithmic}
\end{algorithm}

\paragraph{Analysis of the Competitive Ratio.}
The analysis follows the proof of Proposition~\ref{prop:intervals}.
Recall that $\sigma(a)$ denotes the subsequence of $\sigma$ consisting
of ranges $r_i$ that are unstabbed upon arrival and stabbed initially by a point
colored $a$.

\begin{lemma}\label{lem:LB-algp}
 The ranges in $\sigma(a)$ are pairwise disjoint.
\end{lemma}

\begin{proof}
  Assume for the sake of contradiction that $r_i,r_j\in\sigma(a)$ and
  $z\in r_i\cap r_j$.  Let $[a_i,b_i]$ denote the endpoints of the
  interval $I_i=r_i\cap V$, and define $[a_j,b_j]$ and $I_j$
  similarly.  The proof of  Proposition~\ref{prop:intervals} proves that
  $I_i\cup I_j$ is not an interval.  This implies that $z\not\in V$ and
  that there is an extreme point $t\in V$ between $I_i$ and $I_j$.

  Without loss of generality, $b_i < t < a_j$. Let $(z)_x$ denote the
  $x$-coordinate of point $z$.  Assume that $(z)_x\leq (t)_x$ (the
  other case is handled similarly).  See Fig.~\ref{fig:points} for
  an illustration. Let $L_j$ denote a line that induces the range $r_j$,
  i.e., the set of points below $L_j$ is $r_j$.  Let $L_t$ denote a
  line that separates $t$ from $X\setminus \{t\}$, i.e., $t$ is the
  only point below $L_t$.  Then, $L_t$ passes below $z$, above $t$,
  and below $a_j$.  On the other hand, $L_j$ passes above $z$, below
  $t$, and above $a_j$.  Since $(z)_x \leq (t)_x < (a_j)_x$, it
  follows that the lines $L_t$ and $L_j$ intersect twice, a
  contradiction, and the lemma follows.
\end{proof}
\begin{figure}
\centering
\includegraphics[width=0.6\textwidth]{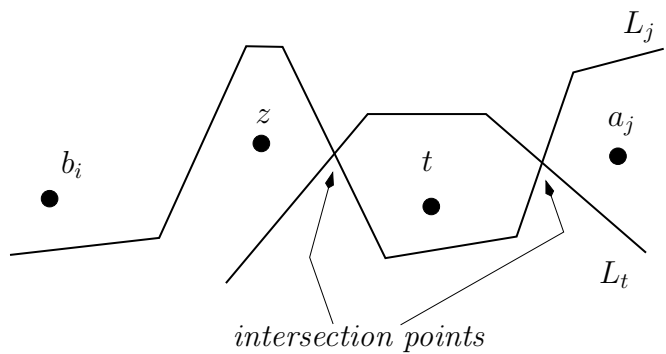}
\caption{Proof of Lemma~\ref{lem:LB-algp}. The lines $L_j$ and $L_t$  are depicted
  as  polylines only for the purpose of depicting their above/below
  relations with the points.}
 \label{fig:points}
 \end{figure}

Lemma~\ref{lem:LB-algp} implies that $|\opt(\sigma)| \geq \max_{a}
|\sigma(a)|$.  On the other hand $|\algp(\sigma)| = \sum_a |\sigma(a)|
\leq (1+\log n)\cdot \max_{a} |\sigma(a)|$, and
Theorem~\ref{thm:UB-points} follows.

\section{Points and Unit Discs}\label{sec:discs}
In this section we consider a special instance of the online hitting
set problem in which the ground set $X$ is a finite set of $n$ points
in $\RR^2$.  The ranges are subsets of points that are contained in a
unit disc.  Formally, a unit disc $d$ centered at $o$ is the set
$d\eqdf\{x\in \RR^2 : ||x-o||_2 \leq 1\}$.  The range $r(d)$ induced
by a disc $d$ is the set $r(d)\eqdf \{x \in X : x\in d\}$.  The circle
$\partial d$ is defined by $\partial d\eqdf\{x\in \RR^2 : ||x-o||_2 = 1\}$.

As in the case of points and half-planes, the hypergraph of points and
unit discs is not $I$-type. To see this, assume that the distances between
the four points in Fig.~\ref{fig:pts-planes} are small. In this case,
the lines $L_1$ and $L_2$ can be replaced by unit discs that induce
the same ranges.

\begin{theorem}\label{thm:LB-discs}
The competitive ratio of every online hitting set algorithm for
points and unit discs is $\Omega(\log n)$.
\end{theorem}
\begin{proof}
  Reduce an instance of intervals on a line to points and half-planes
  as follows.  Position $n$ points on a line such that the distance
  between the first and last point is less than one. For each
  interval, there exists a unit disc that intersects the points in
  exactly the same points as the interval. Thus, the lower bound for
  points and intervals (Prop.~\ref{prop:intervals}) holds also for
  unit discs.
\end{proof}

\begin{theorem}\label{thm:UB-discs}
There exists an online hitting set algorithm for points and
discs that achieves a competitive ratio of $O(\log n)$.
\end{theorem}

\subsection{Proof of Theorem~\ref{thm:UB-discs}}

\paragraph{Partitioning.}
We follow Chen et al.~\cite{chen2009online} with the following
partitioning of the plane (see Fig.~\ref{fig:tile}).  Partition
the plane into square tiles with side-lengths $1/2$.  Consider a
square $s$ in this tiling.  Let $S$ denote a square concentric
with $s$ whose side length is $5/2$.  Partition $S$ into four
quadrants, each a square with side length $5/4$. Let $S^i$ denote
a quadrant of $S$ and let $o^i$ denote its center, for $i \in
\{1,2,3,4\}$.
Let $D_s$ denote the set of unit discs $d$ such that $d\cap s \neq \emptyset$.

\begin{proposition}\label{tile}
  If $d\in D_s$, then $d\cap \{o^1,\ldots,o^4\} \neq \emptyset$.
\end{proposition}

For $d\in D_s$, let $\tau(s,d) \eqdf \min\{i : o^i \in d\}$.  For
$\tau\in \{1,\ldots,4\}$, let $D_{s,\tau}$ denote the set $\{d\in
D_s \mid \tau(s,d) = \tau, d\cap s \cap X\neq \emptyset\}$. The
following lemma shows that circles bounding the discs in
$D_{s,\tau}$ behave like pseudo-lines when restricted to a
subregion of $S$.
\begin{lemma}[\cite{chen2009online}]\label{lem:CKS}
  Let $K^{s,\tau}$ denote the convex cone with apex $o^\tau$ spanned by
  $s$. Then, for any pair of discs $d,d'\in D_{s,\tau}$, the circles
  $\partial d$ and $\partial d'$ intersect at most once in $K^{s,\tau}$.
\end{lemma}

\begin{figure}
\centering
\includegraphics[width=0.6\textwidth]{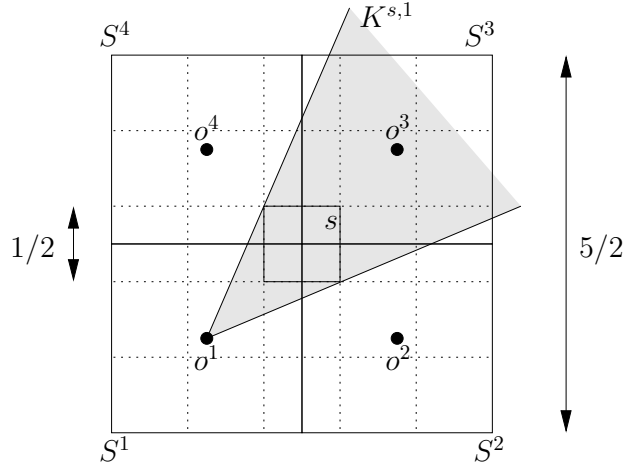}
\caption{A partitioning of the plane from Chen et al.~\cite{chen2009online}}
 \label{fig:tile}
 \end{figure}

\paragraph{Extreme points.}
For every square tile $s$ and every $\tau\in\{1,\ldots,4\}$, we define
a set $V_{s,\tau}$ of \emph{extreme points} as follows.
\[
V_{s,\tau} \eqdf \{ x \in X | \exists d\in D_{s,\tau}: d\cap s \cap X = \{x\}\}.
\]
Note that if $d\in D_{s,\tau}$ and $d\cap s \cap X\neq\emptyset$,
then $d\cap V_{s,\tau} \neq \emptyset$.

Let $\theta_{s,\tau}: V_{s,\tau} \rightarrow [0,2\pi]$ denote an
\emph{angle function}, where $\theta_{s,\tau}(x)$ equals the slope of
the line $o^\tau x$.
Let $\{p_i\}_{i=1}^{|V_{s,\tau}|}$ denote an ordering of $V_{s,\tau}$
in increasing $\theta_{s,\tau}$ order.  For a disc $d\in D_{s,\tau}$,
we say that $d\cap V_{s,\tau}$ is an \emph{interval} if there exist
  $i,k$ such that $d\cap V_{s,\tau}= \{ p_j \mid i\leq j \leq k\}$.
\begin{proposition}[\cite{cheilaris2010choosability}]\label{prop:angle}
  The angle function $\theta_{s,\tau}$ is one-to-one, and $d\cap
  V_{s,\tau}$ is an interval, for every $d\in D_{s,\tau}$.
\end{proposition}

\paragraph{Vertex ranking.}
Let $P_{s,\tau}$ denote the path graph over $V_{s,\tau}$ where $p_i$
is a neighbor of $p_{i+1}$ for $i = 1,\ldots, |V_{s,\tau}|-1$.  Let
$c^{s,\tau}: V_{s,\tau}\rightarrow \NN$ denote a vertex ranking with
respect to $P_{s,\tau}$ that uses $\lfloor \log_2 (2|V_{s,\tau}|)
\rfloor$ colors.

Consider a disc $d\in D_{s,\tau}$. Let $r=r(d)$ denote the range
$d\cap X$.  Assume that $r\cap s\neq \emptyset$.  Let
$c_{\max}^{c,\tau} (r) \eqdf \max \{ c^{s,\tau} (v) \mid v\in r\cap
V_{s,\tau}\}$. Let $v_{\max}^{s,\tau} (r)$ denote the vertex $v\in r\cap V_{s,\tau}$
such that $c^{s,\tau}(v)=c_{\max}^{s,\tau}(r)$.

\subsubsection{Algorithm Description}
A listing of the algorithm \algd\ appears as Algorithm~\ref{alg:d}.
The algorithm requires the following preprocessing:
\begin{inparaenum}[(i)]
\item Compute a tiling of the plane with $1/2\times 1/2$ squares.
  Each point $x\in X$ must lie in the interior of a tile. This is easy
  to achieve since $X$ is finite.
\item For every tile $s$, compute the four types of extreme points
  $V_{s,\tau}$, and order each $V_{s,\tau}$ in increasing
  $\theta_{s,\tau}$ order.
\item Compute a vertex ranking $c^{s,\tau}$ for each $V_{s,\tau}$.
\end{inparaenum}
The algorithm maintains a hitting set $C_i$ of the $i-1$ ranges
$\{r_1,\ldots,r_{i-1}\}$ that have been input so far.  Upon
arrival of a range $r_i=r(d_i)$, if it is stabbed by $C_{i-1}$,
then simply update $C_i\gets C_{i-1}$.  Otherwise, a vertex
$v_{i,s}$ is selected from each square tile $s$ such that $r_i
\cap s \neq \emptyset$. These vertices are added to $C_{i-1}$ to
obtain $C_i$.

Lemma~\ref{lem:CKS} provides an interpretation of Algorithm \algd\ as
a reduction to the case of hitting subsets of points below a
pseudo-line (i.e., pseudo half-planes). Each square tile $s$ and type
$\tau\in\{1,\ldots,4\}$ defines an instance of points and pseudo
half-planes with respect to the set $X_s\eqdf X\cap s$ of points and
the subsets $d\cap X_s$ for discs $d\in D_{s,\tau}$. The algorithm
maintains a different invocation of \algp\ for each square $s$ and
type $\tau$. Upon arrival of an unstabbed disc $d$, the algorithm
inputs the range $d\cap s \cap X$ to each instance of \algp\
corresponding to a square $s$ and a type $\tau$ such that $d\in
D_{s,\tau}$.

\begin{algorithm}
  \caption{\algd$(X)$ - an online  hitting set for unit discs.}
\label{alg:d}
\begin{algorithmic}[1]
  \REQUIRE $X\subset \RR^2$ is a set of $n$ points. A tiling by
  $1/2\times 1/2$ squares. Four types of extreme points $V_{s,\tau}$
  per tile. A vertex ranking $c^{s,\tau}$ of $V_{s,\tau}$ with respect
  to the ``angular'' order.  \STATE $C_0\gets \emptyset$
  \FOR[arrival of a range $r_i=r(d_i)$]{$i = 1$ to $\infty$}
  \IF{$r_i$ not stabbed by $C_{i-1}$}
\FORALL{square tiles $s$ such that $r_i\cap s\neq\emptyset$}
\STATE $\tau \gets \tau(s,d_i)$ \COMMENT{find the type of $d_i$ wrt $s$}
\STATE $v_{s,i} \gets v^{s,\tau}_{\max} (r_i \cap V_{s,\tau})$ \COMMENT{find the vertex with the max color}
\STATE  $C_i\gets C_{i-1} \cup \{v_{s,i}\}$
\ENDFOR
\ELSE \STATE $C_i\gets C_{i-1}$
\ENDIF \ENDFOR
\end{algorithmic}
\end{algorithm}

\subsubsection{Analysis of The Competitive Ratio}

Let $\sigma=\{r_i\}_i$ denote the input sequence. Let
$\sigma^A\subseteq \sigma$ denote the subsequence of ranges $r_i$ such
that $r_i$ is unstabbed upon arrival (i.e., $r_i$ is not stabbed by
$C_{i-1}$).
\begin{proposition}\label{prop:25}
$|\algd(\sigma)| \leq 16 \cdot |\sigma^A|$.
\end{proposition}
\begin{proof}
  Each disc intersects at most $16$ square tiles. Upon arrival of an
  unstabbed disc, at most one point is added to the hitting set, for
  each intersected square.
\end{proof}

The following lemma shows that, if two discs contain a common point
$x\in s$, are of the same type $\tau$, and are unstabbed upon arrival,
then they are stabbed by extreme points in $V_{s,\tau}$ of different
colors.
\begin{lemma}\label{lem:distinct}
  If $x\in X\cap s$, $r_i,r_j \in \sigma^A \cap
  D_{s,\tau}$ and $x\in r_i\cap r_j$, then $c^{s,\tau}(v_{s,i}) \neq c^{s,\tau}(v_{s,j})$.
\end{lemma}
\begin{proof}
  To shorten notation let $c=c^{s,\tau}$, $V=V_{s,\tau}$, and
  $\theta=\theta^{s,\tau}$.  Assume for the sake of contradiction that
  $c(v_{s,i}) = c(v_{s,j})$. By Prop.~\ref{prop:angle}, $r_i\cap V$ is
  an interval, which we denote by $I_i = [a_i,b_i]$.  Similarly,
  $I_j=[a_j,b_j]$ is the interval for $r_j\cap V$.  Since $c^{s,\tau}$
  is a unique-max coloring of the intervals in $V_{s,\tau}$, $I_i\cup
  I_j$ is not an interval, so there must be an extreme point in between
  the intervals. Denote this in between point by $t$. Without loss of
  generality, $\theta(b_i)<\theta(t)<\theta(a_j)$. Assume that
  $\theta(z)\leq \theta(t)$. Consider a disc $d_j\in D_{s,\tau}$ such
  that $r_j=r(d_j)$. Consider a disc $d_t\in D_{s,\tau}$ such that
  $d_t\cap X_s = \{t\}$. We claim that the circles $\partial d_j$ and
  $\partial d_t$ intersect twice in the cone $K^{s,\tau}$,
  contradicting Lemma~\ref{lem:CKS}. Indeed, $\partial d_t$ passes
  ``below'' $x$, ``above'' $t$, and ``below'' $a_j$.  On the other
  hand, $\partial d_j$ passes above $x$, below $t$, and above $a_j$.
  The case $\theta(z)> \theta(t)$ is proved similarly by considering
  the discs $d_t$ and $d_i$.
\end{proof}

Let $\sigma(x)$ denote the subsequence of ranges $r_i$ such that $x\in
r_i$. The following lemma proves that the algorithms stabs a sequence
of discs that share a common point by $O(\log n)$ points.
\begin{lemma}\label{lem:local}
For every $x\in X$, $|\algd (\sigma(x))| \leq 64 \cdot \lfloor
\log_2 (2n) \rfloor$.
\end{lemma}
\begin{proof}
  Fix a point $x\in X$, and let $s$ denote the tile such that $x\in
  s$. Let $\sigma^A(x)$ denote the sequence of ranges in $\sigma(x)$
  that were unstabbed upon arrival in an execution of
  $\alg(\sigma(x))$. By Prop.~\ref{prop:25}, $|\algd (\sigma(x))| \leq
  16\cdot |\sigma^A(x)|$.

  The disc $d_i$ of each range $r_i\in \sigma^A(x)$ belongs to one of
  four types $D_{s,\tau}$, for $1\leq\tau\leq 4$ By<
  Lemma~\ref{lem:distinct}, the ranges in $\sigma^A(x) \cap D_{s,\tau}$ are
  stabbed by extreme points in $V_{s,\tau}$, the colors of which are
  distinct. Each vertex ranking $c^{s,\tau}$ uses at most $\lfloor
  \log_2 (2n) \rfloor$ colors. Thus, $|\sigma^A(x)| \leq
  \sum_{\tau=1}^4 |\sigma^A(x)\cap V_{s,\tau}| \leq 4\cdot
  \lfloor \log_2 (2n) \rfloor$, and the lemma follows.
\end{proof}

\begin{proof}[Proof of Theorem~\protect\ref{thm:UB-discs}] Consider an
execution of $\algd(\sigma)$ and independent
executions of $\algd(\sigma(x))$, for every $x\in \opt(\sigma)$.
Every time $\algd(\sigma)$ is input an unstabbed range $r_i$, at least
one of the executions of $\algd(\sigma(x))$ is also input $r_i$, and
$r_i$ is also unstabbed upon arrival. This implies that
$|\algd(\sigma)| \leq \sum_{x\in \opt(\sigma)} |\algd(\sigma(x))|$.

By Lemma~\ref{lem:local}, $ |\algd(\sigma(x))| = O(\log n)$. This
implies that $|\algd(\sigma)| = O(\log n) \cdot |\opt(\sigma)|$, and
the theorem follows.
\end{proof}

\section{Discussion}
We would like to suggest two open problems.
\begin{enumerate}
\item Design an online hitting set algorithm for points and arbitrary
  discs, the competitive ratio of which is $o(\log^2 n)$ or prove a lower bound of $\Omega(\log^2 n)$.
\item Design an online hitting set algorithm with a logarithmic
  competitive ratio for any hypergraph with bounded VC-dimension or obtain a lower bound as above. Alon
  et al. obtain an $O(\log^2 n)$ competitive ratio, and the best
  known lower bound is $\Omega(\log n)$.
\end{enumerate}

\section*{Acknowledgments}
We would like to thank Nabil Mustafa and Saurabh Ray for discussions
on the open problems suggested in the last section.

\bibliographystyle{alpha}
\bibliography{hitting-sets-online}

\newcommand{\etalchar}[1]{$^{#1}$}
\begin{thebibliography}{BMKM05}

\bibitem[AAA{\etalchar{+}}09]{alon2009online}
N.~Alon, B.~Awerbuch, Y.~Azar, N.~Buchbinder, and J.S. Naor.
\newblock {The Online Set Cover Problem}.
\newblock {\em SIAM Journal on Computing}, 39:361, 2009.

\bibitem[BEY98]{borodin1998online}
A.~Borodin and R.~El-Yaniv.
\newblock {\em {Online computation and competitive analysis}}, volume~2.
\newblock Cambridge University Press Cambridge, 1998.

\bibitem[BG95]{bronnimann1995almost}
H.~Br{\"o}nnimann and M.T. Goodrich.
\newblock {Almost optimal set covers in finite VC-dimension}.
\newblock {\em Discrete and Computational Geometry}, 14(1):463--479, 1995.

\bibitem[BMKM05]{ben2005constant}
B.~Ben-Moshe, M.J. Katz, and J.S.B. Mitchell.
\newblock {A constant-factor approximation algorithm for optimal terrain
  guarding}.
\newblock In {\em Proceedings of the sixteenth annual ACM-SIAM symposium on
  Discrete algorithms}, pages 515--524. Society for Industrial and Applied
  Mathematics, 2005.

\bibitem[Chv79]{chvatal1979greedy}
V.~Chvatal.
\newblock {A greedy heuristic for the set-covering problem}.
\newblock {\em Mathematics of operations research}, 4(3):233--235, 1979.

\bibitem[CKS09]{chen2009online}
K.~Chen, H.~Kaplan, and M.~Sharir.
\newblock {Online conflict-free coloring for halfplanes, congruent disks, and
  axis-parallel rectangles}.
\newblock {\em ACM Transactions on Algorithms (TALG)}, 5(2):1--24, 2009.

\bibitem[CS10]{cheilaris2010choosability}
P.~Cheilaris and S.~Smorodinsky.
\newblock {Choosability in geometric hypergraphs}.
\newblock {\em Arxiv preprint arXiv:1005.5520}, 2010.

\bibitem[CV07]{clarkson2007improved}
K.L. Clarkson and K.~Varadarajan.
\newblock {Improved approximation algorithms for geometric set cover}.
\newblock {\em Discrete and Computational Geometry}, 37(1):43--58, 2007.

\bibitem[ELRS03]{even2003conflict}
G.~Even, Z.~Lotker, D.~Ron, and S.~Smorodinsky.
\newblock {Conflict-Free Colorings of Simple Geometric Regions with
  Applications to Frequency Assignment in Cellular Networks}.
\newblock {\em SIAM Journal on Computing}, 33:94, 2003.

\bibitem[ERS05]{even2005hitting}
G.~Even, D.~Rawitz, and S.M. Shahar.
\newblock {Hitting sets when the VC-dimension is small}.
\newblock {\em Information Processing Letters}, 95(2):358--362, 2005.

\bibitem[Fei98]{feige}
U.~Feige.
\newblock {A threshold of ln n for approximating set cover}.
\newblock {\em Journal of the ACM (JACM)}, 45(4):634--652, 1998.

\bibitem[HM85]{hochbaum1985approximation}
D.S. Hochbaum and W.~Maass.
\newblock {Approximation schemes for covering and packing problems in image
  processing and {VLSI}}.
\newblock {\em Journal of the ACM (JACM)}, 32(1):136, 1985.

\bibitem[IRV88]{iyer1988optimal}
A.V. Iyer, H.D. Ratliff, and G.~Vijayan.
\newblock {Optimal node ranking of trees.}
\newblock {\em Information Processing Letters}, 28(5):225--229, 1988.

\bibitem[Joh74]{johnson1974approximation}
D.S. Johnson.
\newblock {Approximation algorithms for combinatorial problems*}.
\newblock {\em Journal of Computer and System Sciences}, 9(3):256--278, 1974.

\bibitem[Kar72]{Karp72}
R.M. Karp.
\newblock {Reducibility among combinatorial problems, 85--103}.
\newblock In {\em Proc. Sympos., IBM Thomas J. Watson Res. Center, Yorktown
  Heights, NY}, 1972.

\bibitem[KMS95]{katchalski1995ordered}
M.~Katchalski, W.~McCuaig, and S.~Seager.
\newblock {Ordered colourings}.
\newblock {\em Discrete Mathematics}, 142(1-3):141--154, 1995.

\bibitem[KR99]{anil1999covering}
V.S.A. Kumar and H.~Ramesh.
\newblock {Covering rectilinear polygons with axis-parallel rectangles}.
\newblock In {\em Proceedings of the thirty-first annual ACM symposium on
  Theory of computing}, pages 445--454. ACM, 1999.

\bibitem[Sch89]{Schaffer89}
Alejandro~A. Sch{\"a}ffer.
\newblock Optimal node ranking of trees in linear time.
\newblock {\em Information Processing Letters}, 33(2):91--96, 1989.

\bibitem[Smo12]{CF-survey}
S.~Smorodinsky.
\newblock Conflict-free coloring and its applications.
\newblock In I.~Barany, K.J. Boroczky, G.~Fejes~Toth, and J.~Pach, editors,
  {\em Geometry-Intuitive, Discrete, and Convex}, Bolyai Society Mathematical
  Studies. Springer, 2012.

\bibitem[ST85]{sleator1985amortized}
D.D. Sleator and R.E. Tarjan.
\newblock {Amortized efficiency of list update and paging rules}.
\newblock {\em Communications of the ACM}, 28(2):202--208, 1985.

\end{thebibliography}

\appendix

\end{document}